\documentclass[trans]{IEEEtran}
\IEEEoverridecommandlockouts
\usepackage{cite}
\usepackage{amsthm}
\usepackage{amsmath,amssymb,amsfonts}
\usepackage{graphicx}
\usepackage{textcomp}
\usepackage{xcolor}
\usepackage{algorithm}
\usepackage{algorithmicx}
\usepackage{algpseudocode} 
\usepackage{bm}
\newtheorem{lemma}{Lemma}
\newtheorem{thm}{Theorem}
\def\BibTeX{{\rm B\kern-.05em{\sc i\kern-.025em b}\kern-.08em
    T\kern-.1667em\lower.7ex\hbox{E}\kern-.125emX}}
\begin{document}

\title{Nonuniform Quantized Decoder for Polar Codes with Minimum Distortion Quantizer\\
}

\author{Zhiwei~Cao, Hongfei~Zhu, Yuping~Zhao, Dou~Li \\ 
	School of Electronics Engineering and Computer Science\\
	Peking University, Beijing, 100871, China\\
	Email:\{cao\_zhiwei, zhuhongfei, yuping.zhao, lidou\}@pku.edu.cn
	\thanks{This work was financially supported in part by the National Key Research and Development Program of China under the Grant No. 2020YFB1807802, 2016ZX03001018-005 and in part by Huawei Technologies Co., Ltd.}}

\maketitle

\begin{abstract}
We propose a nonuniform quantized decoder for polar codes. The design metric of the quantizers is to minimize the distortion incurred by quantization. The quantizers are obtained via dynamic programming and the optimality of the quantizer is proved as well. Simulation results show that  the error correction performance degradation of the proposed nonuniform quantized decoder is less than 0.1 dB compared to conventional float-point decoders under 5-bit nonuniform quantization, outperforming previous uniform quantized decoders for polar codes significantly.
\end{abstract}

\section{Introduction}
\IEEEPARstart{P}{olar} codes\cite{ArikanPolar} are the first error-correcting codes which can achieve symmetric capacity of memoryless channls with an explicit structure and low decoding complexity. Traditional float-point decoders for polar codes, namely successive cancellation (SC) decoding\cite{ArikanPolar} and successive cancellation list (SCL) decoding\cite{SCL} require a huge amount of float point arithmetics, which increases the hardware complexity and energy consumption drastically. 

It is well-known that quantization is a critical issue in hardware implementation of decoders. Many quantization methods have been proposed in order to implement low complexity decoders for polar codes.  In \cite{RobustQuantization}, the authors analyze the robustness of polar codes with a quantized decoder and found that polar codes are robust to quantization. In \cite{UniformQuantizationSCL}, a uniform quantizer with different design metric such as minimizing squared errors is developed. As a result, 6-bit uniform quantized SC decoder yields near float-point performance.  An adaptive quantizer is proposed in \cite{TerabitsPerSecondDecoder} and achieve terabits per second throughputs. 

The quantization methods mentioned above all deal with the quantization of SC decoder, which is less practical in real world. In \cite{QuantizationCASCL}, uniform quantization of channel log-likelihood ratio (LLR) and internal LLRs are studied. The authors claimed that 6-bit quantization for channel LLRs and 8-bit quantization for internal LLRs yields negligible performance loss. In order to study the error correction performance of SCL decoders under very coarsely quantized cases, the authors in \cite{TernaryQuantized} focused on the case where the decoder messages are quantized into 3 levels. Lookup tables are designed heuristically for performing SC decoding inside the decoder. 

Exising works on the design of the quantized decoder for polar codes mainly focus on uniform quantization or designing the quantization mapping of which the optimality is not guaranteed. In this letter, motivated by the fact that conventional float-point decoders for polar codes are all based on LLRs, we proposed a nonuniform quantized decoder for polar codes with minimum distortion quantizers for both channel LLRs and internal LLRs, aiming to minimize the squared error, which is similar to the case studied in the Lloyd algorithm\cite{lloyd}. However, some stark differences distinguish the proposed approach from Lloyd algorithm. First of all, we aim to design a quantizer for a discrete random variable instead of a continuous random variable. Secondly, the principal of proposed quantizer design algorithm is to solve a combinatorial optimization problem via dynamic programming instead of alternating minimizing used in Lloyd algorithm. Most importantly, it is neccssary to obtain the analytical expression of the probability distribution of internal LLRs if we wish to apply Lloyd algorithm to design quantizers, which is intractable in the context of polar codes. 

 We first provide some background knowledge in Section \ref{Section_Preliminaries}. Then we present the system model in Section \ref{Section_SystemModel}. Subsequently, we define the concept of quantized density evolution with the aid of quantizers in Section \ref{QuantizedDensityEvolution}. Later on, the proposed nonuniform quantizer for discrete LLR distribution that minimizes distortion is derived in Section \ref{MMSEQuantizer} and the  optimality of the proposed minimum distortion quantizer is discussed as well. SC and SCL decoding processing using the precomputed quantizers are discussed in Section \ref{QDecode} in the following. Simulation results are illustrated in Section \ref{Simulation} and concluding remarks are given in Section \ref{conclusion} ultimately.
\section{Preliminaries}\label{Section_Preliminaries}
\subsection{Polar Codes}
To construct a $(N, K)$ polar code $\mathcal{P}(N, K)$, $K$ information bits $u_{\mathcal{A}}$ and $N-K$ frozen bits $u_{\mathcal{A}^c}$ are assigned to the $K$ reliable bits and $N-K$ unreliable bits respectively, where $\mathcal{A}$ is the information bits set and  $\mathcal{A}^c$ is the frozen bits set. $N - K$ frozen bits are set to the predefined value. The codeword can be obtained by multiplying source bits $u_1^N$ with generator matrix $G_N$
\begin{align}
x_1^N = u_1^NG_N = u_1^NF_2^{\otimes n}
\end{align}

where, $n = \log_2N$, $F_2^{\otimes n}$ is the $n$-th Kronecker product of the polarization matrix 
\begin{align*}
{F}_2 = \left[ {\begin{array}{*{20}{c}} 
	1&0\\
	1&1
	\end{array}} \right] 
\end{align*}
\vspace{-0.5cm}
\subsection{Successive Cancellation Decoding}\label{Pre-SC}
Denote the LLR of the channel outpt $y_j$ as $\alpha_{0, j}$ and the $j$-th internal LLR in the $i$-th decoding stage is $\alpha_{i, j}$, $i=1,2,...,n=\log_2N, j=1,2,...,N$. In \cite{ArikanPolar}, the LLR of a message bit is defined as  
\begin{equation}\label{LLR_message_bits}
	\alpha_{n, i} = \ln \frac{P(\bm{y}, \hat{u}_1, \hat{u}_2, ..., \hat{u}_{i-1}|u_i=0)}{P(\bm{y},  \hat{u}_1, \hat{u}_2, ..., \hat{u}_{i-1}|u_i=1)}
\end{equation} 
For SC decoding, internal LLRs can be computed in the following recursive and hardware-friendly way
\begin{align} 
	\label{f} \alpha_{i, j} = \text{sign}(\alpha_{i - 1, j})\cdot\text{sign}(\alpha_{i, j + N/2^i})\cdot \\ \notag \min\{|\alpha_{i - 1, j}|, |\alpha_{i - 1, j + N/2^i}|\}\\
	\label{g} \alpha_{i, j+N/2^i} = (-1)^{\hat{u}_{i, j}}\cdot \alpha_{i - 1, j} + \alpha_{i - 1, j+N/2^i}
\end{align}
where $\hat{u}_{i, j}$ is the hard decision result for internal bit $u_{i, j}$. (\ref{f}) and (\ref{g}) are often referred as $f$ and $g$ function in the literature.

 When it comes to decoding a bit, according to its LLR and whether it is a frozen bit, using the following rule to get the estimation
\begin{align}\label{SCDecision}
\hat{u}_i =\begin{cases}
0,  \ \  &\text{if} \  \alpha<0 \ \text{or} \ i\in\mathcal{A}^c \\
1,  \ \  &\text{otherwise}
\end{cases}	
\end{align}
\vspace{-0.5cm}
\subsection{Successive Cancellation List Decoding} \label{Pre-SCL}
SCL decoding\cite{SCL} constains $L$ concurrent SC decoders . Improved error correction performance is available by saving multiple different decoding paths and selecting one best path after the depth-first traversal.
The quality of a decoding path is evaluated by path metric (PM)\cite{LLR-SCL}.  At leaf nodes $i$, $\hat{u}_i$ is estimated as 0 or 1, and the PM of these two results are updated as 
\begin{align}\label{PM-HWF}
\begin{split}
\text{PM}_{0}^l &= 0 \\
\text{PM}_{i}^l &=\begin{cases}
\text{PM}^l_{i-1} + |\alpha_i^l|, \ \  &\text{if} \ \hat{u}_i^l \neq \frac{1}{2}(1-\text{sgn}(\alpha_i^l)), \\ 
\text{PM}^l_{i-1},  \ \  &\text{otherwise}
\end{cases}
\end{split}
\end{align}
Among the $2L$ candidates, $L$ candidates with lowest PM values are selected in order to avoid the exponential increasing complexity.
\vspace{-0.5cm}
\section{System Model}\label{Section_SystemModel}
The system model in this letter is illustrated in Fig. \ref{SystemModel}. The message vector in the transmitter $u_{\mathcal{A}}$ is first encoded by the nonsystematic polar encoder to obtain the codeword $x_1^N$. Code construction is accomplished by beta-expansion\cite{BetaExpansion}. The codeword is modulated with binary phase shift keying (BPSK) modulator, which yields symbol vector $s_1^N$. The modulated symbols are transmitted through the additive white Gaussian noise (AWGN) channel with noise variance $\sigma^2$. For the AWGN channel, the received soft value $y$ is first transformed to LLR soft value $m$ by the LLR converter in the receiver
\begin{align}
	m = LLR(y) = \ln\frac{p(y|u=0)}{p(y|u=1)} ={\sigma^2}y
\end{align}

Note that the conditional distribution $p(m|x=+1)$ and $p(m|x=-1)$ are Gaussian distribution with mean value $\pm2/\sigma^2$ respectively and variance $4/\sigma^2$. Therefore the unconditioned distribution of LLR is a bimodal  Gaussian distribution. We first quantize channel LLRs $m$ uniformly with a large quantization level. This discrete distribution will be the input of the quantizer design algorithm in Section \ref{QuantizerDesign}, which produces the minimum distortion quantizer. The uniform quantizer concatenated with the minimum distortion quantizer yields the LLR quantizer in Fig. \ref{SystemModel}.

With the real number $m$ as input, the LLR quantizer produces its quantized value $\tilde{m}$. Finally, the quantized symbol vector $\tilde{m}_1^N$ is fed to the quantized decoder, which yields the estimation of the message vector $\hat{u}_{\mathcal{A}}$. 

 In this letter, we fix the uniform quantization level to 128 through massive simulation and find it yields best trade-off between decoding performance and computational complexity. We keep the quantization level for the channel LLR compression quantizer and the decoder quantizers the same in simulations. We use $n$-bit quantization to represent the quantization level is $2^n$ in the rest of this paper. 
\begin{figure}[t!]
	\centering
	\includegraphics[width=.4\textwidth]{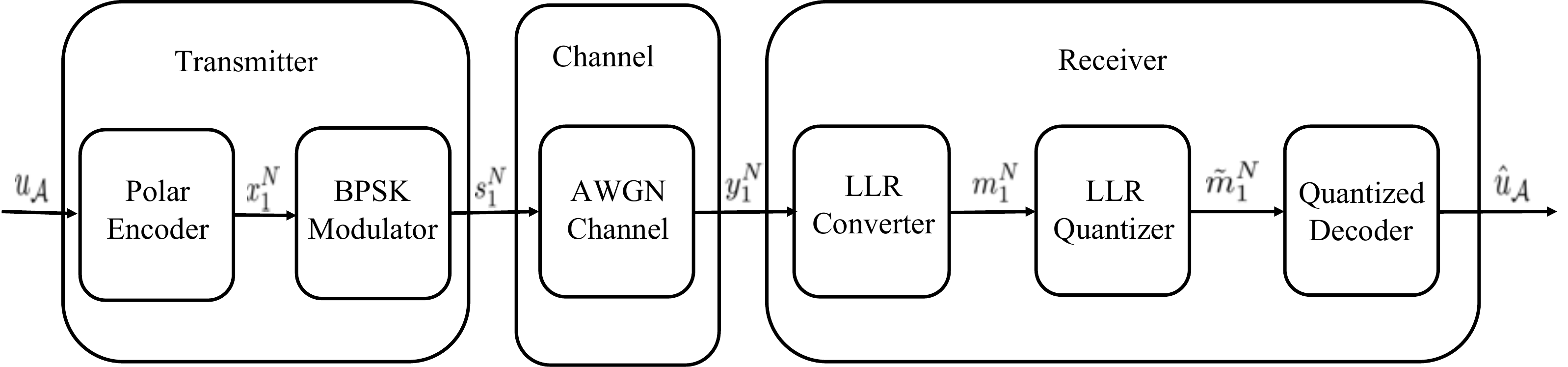}
	\caption{System model in this letter.} 
	\label{SystemModel} 
	\vspace{-0.5cm}
\end{figure}
\vspace{-0.3cm}
\section{Quantized Decoding For Polar Codes}\label{Body}
In this section, we first describe the process of quantized density evolution and the necessity of quantizers. Subsequently, we delineate the proposed minimum distortion quantizer. Utimately, we briefly describe the quantized SC  and SCL decoding with the minimum distortion quantizers.
\vspace{-0.5cm}
\subsection{Quantized Density Evolution}\label{QuantizedDensityEvolution}
Denote the quantized channel output LLRs and internal LLRs as $\tilde{\alpha}_{0, i}$ and $\tilde{\alpha}_{i, j}$ respectively. From the recursive structure of polar codes, one can find that
\begin{align}
	\tilde{\alpha}_{i, j} = \mathcal{Q}_f(f(\tilde{\alpha}_{i-1, j}, \tilde{\alpha}_{i-1, j+N/2^i}))
\end{align}
\begin{align}
	\tilde{\alpha}_{i, j+N/2^i} = \mathcal{Q}_g(g(\tilde{\alpha}_{i-1, j}, \tilde{\alpha}_{i-1, j+N/2^i}, \hat{u}_{i, j}))
\end{align}
where $\mathcal{Q}_f$ and $\mathcal{Q}_g$ represent the deterministic quantizers taking the place of the float-point $f$ and $g$ arithmatic. Because those LLRs from the AWGN channel, i.e. $\alpha_{0,j}$, are independent Gaussian random variables, the internal LLRs at each decoding stage remain independent due to the recursive structure of polar codes. Therefore,  the distribution of  $\tilde{\alpha}_{i,j}$ and $\tilde{\alpha}_{i, j+N/2^i}$ can be computed through  combination
\begin{small}
\begin{align}\label{QuantizedDensityF}
P(\tilde{\alpha}_{i, j} = e) = \sum_{f(n, k)=e} P(\tilde{\alpha}_{i-1, j} = n)P(\tilde{\alpha}_{i-1, j+N/2^i} = k)
\end{align} 
\end{small}
\begin{small}
\begin{align}\label{QuantizedDensityG}
\begin{split}
P(\tilde{\alpha}_{i, j+N/2^i} = e) = \sum_{g(n, k, \hat{u}_{i, j})=e} P(\hat{u}_{i, j})P(\tilde{\alpha}_{i-1, j} = n)\\P(\tilde{\alpha}_{i-1, j+N/2^i} = k)
\end{split}
\end{align}
\end{small}
where $e\in\mathbb{R}$ is the output of $f$ or $g$ function evalution with the quantized values $n,k\in\mathbb{R}$ as input. Because different inputs may produce the same output for $f$ or $g$ function evaluation, the summation in (\ref{QuantizedDensityF}) and (\ref{QuantizedDensityG}) is necessary. Since the source bits are equiprobable, it is easy to verify that $P(\hat{u}_{i, j}=0)=P(\hat{u}_{i,j}=1)=1/2$. It is obvious that the alphabet size of $\tilde{\alpha}_{i, j}$ will be no less than that of $\tilde{\alpha}_{i, j}$ or $\tilde{\alpha}_{i, j+N/2^i}$. Therefore, if we directly apply (\ref{QuantizedDensityF}) and  (\ref{QuantizedDensityG}) in density evolution without quantization, the alphabet size of the variables after polarization will grow exponentially with respect to the code length, making the density evolution intractable. Hence, we need to design proper quantization scheme to compress the alphabets of $\tilde{\alpha}_{i, j}$ and $\tilde{\alpha}_{i, j+N/2^i}$ so that the quantized density evolution is tractable,  making the probability distributions that are indispensable for the quantizer design algorithm in Section \ref{QuantizerDesign} available as well.
\subsection{Minimum Distortion Quantizer}\label{MMSEQuantizer}
Denote the internal LLR which we want to compress as $L$ with density $p(L)$. Furthermore, we assign a representative value $l_i$ for the $i$-th symbol in the alphabet of $L$. Denote the quantized variable as $T$ with distribution $p(T)$. Likewise, a reconstruction value $t_i$  is assigned to the $i$-th symbol in the alphabet of $T$. Let the alphabet size of $L$ be $M$ and that of $T$ be $K$, which means the quantization level is $K$. The objective of the quantizer $Q$ is to minimize the distortion
\begin{align}
	D = \sum_{k=1}^{K}\sum_{l\in \mathcal{A}(k)}(l-t_k)^2p(l)
\end{align}
where $\mathcal{A}(k)$ is the preimage $(\mathcal{Q}^{-1}(k))$ of the $k$-th quantizer output $t_k$. The set $\mathcal{A}(k)$ and $A(k')$ are disjoint for $k \neq k'$ and the union of all preimages is the alphabet of $L$. The set of all preimages $\{\mathcal{A}(k)\}$ forms a partition for the alphabet of $L$. The optimal quantizer $Q^*$ is obtained through
\begin{align}\label{Obj}
	Q^* = \arg\min_{\{\mathcal{A}(k)\}, \bm{t}} D(\{\mathcal{A}(k)\}, \bm{t})
\end{align}
where $\bm{t} = [t_1,t_2,...,t_K]\in\mathbb{R}^K$. 

When $\bm{t}$ is fixed, the problem is a combinatorial optimization w.r.t. the partition $\{\mathcal{A}(k)\}$ and however, when the partition is fixed, it becomes a convex optimization w.r.t. $\bm{t}$. This observation inspires us to design an efficient dynamic programming algorithm to solve the nontrivial optimization problem (\ref{Obj}).
\subsubsection{Optimal Reconstruction Value}\label{OptQuanta}
Consider the $k$-th quantizer output, if $\mathcal{A}(k)$ is fixed, taking the derivative of $D$ with respect to $t_k$ and set it to zero yield the optimal reconstruction value $t_k^*$
\begin{align}\label{masscenter}
	t_k^* = \frac{\sum_{l\in \mathcal{A}(k)}lp(l)}{\sum_{l\in \mathcal{A}(k)}p(l)}
\end{align}
\subsubsection{Optimal Partition}
As described in Section \ref{OptQuanta}, the necessary condition for the minimum distortion quantizer is that the optimal reconstruction level $t_k$ is computed through (\ref{masscenter}). Thus, the problem for finding the optimal quantizer reduces to finding an optimal partition and then determine each reconstruction level by (\ref{masscenter}) so that the total distortion $D$ is minimized.  The following lemma provides a prerequisite for the optimality of a partition.
\begin{lemma}\label{Lemma1}
	 If the representative values of $L$ follows the ascending order
	\begin{align}\label{order_y}
	l_1 < l_2 < ... < l_M
	\end{align}
	There exists an optimal derministic quantizer $Q^*$ for $L$ so that each $\mathcal{A}(k)$ is a contiguous set of integers 
	\begin{align*}
		\mathcal{A}(k) = \{a_{k-1}+1, a_{k-1}+2, ...., a_{k}\}
	\end{align*}
\end{lemma}
\begin{proof}
	See appendix of \cite{Fisher}.
\end{proof}
The optimization then reduces to find the optimal upper boundary of each $\mathcal{A}(k)$ which satisfies $a_0 = 0 < a_1^* < a_2^* < ... < a_K = M$.
 
  Note that the inequality in Lemma \ref{Lemma1} must be strictly held. If there exist two symbols with the same representative value, namely $l_i = l_{i+1}$, we should first merge these two symbols into a new symbol $l'$  with  representative value $l_i$ and probability $p(l') = p(l_i) + p(l_{i+1})$. The alphabet size of the merged distribution now becomes $M-1$. Without loss of generality, we assume that the representative values always satisfy (\ref{order_y}) in the rest of this letter. 

\subsubsection{Partial Distortion}
We define the partial distortion of the $k$-th quantizer output as the minimum distortion if its preimage in the source symbols ranging from $a_{k-1}+1$ to $a_k$
\begin{align}
	c(a_{k-1}+1\rightarrow a_k) = \sum_{l=a_{k-1}+1}^{a_k}(l-t_k^*)^2p(l)
\end{align}
where, $t_k^*$ is determined via (\ref{masscenter}).
\subsubsection{Quantizer Design  Algorithm}\label{QuantizerDesign}
With all the definition and lemma, we now in a position to provide our quantizer design algorithm for minimizing distortion, which is an instance of dynamic programming. We define the state variable as $S_k(a_k)$, which means the distortion of the optimal quantization of the source symbols 1 to $a_k$ to quantizer output 1 to $k$. It is provable that the following recursive formula holds
\begin{align}\label{Sk_a}
	S_k(a) = \min_{a'} \ \{S_{k-1}(a') + c(a'\rightarrow a)\}
\end{align}
where $a'\in\{k-1, k,...., a-1\}$. Once we obtain $S_K(M)$, the forward computation is accomplished and we find the optimal partition through backtracing subsequently. The quantizer design algorithm is described in detail in Algorithm 1 and we also provide a graphic illustration in the case of $M=5$ and $K=3$ in Fig. \ref{Trellis} for the readers to grasp the idea of our algorithm. What' more, the following theorem guarantees the optimality of the proposed method.
\begin{figure}[t!]
	\centering
	\includegraphics[width=.4\textwidth]{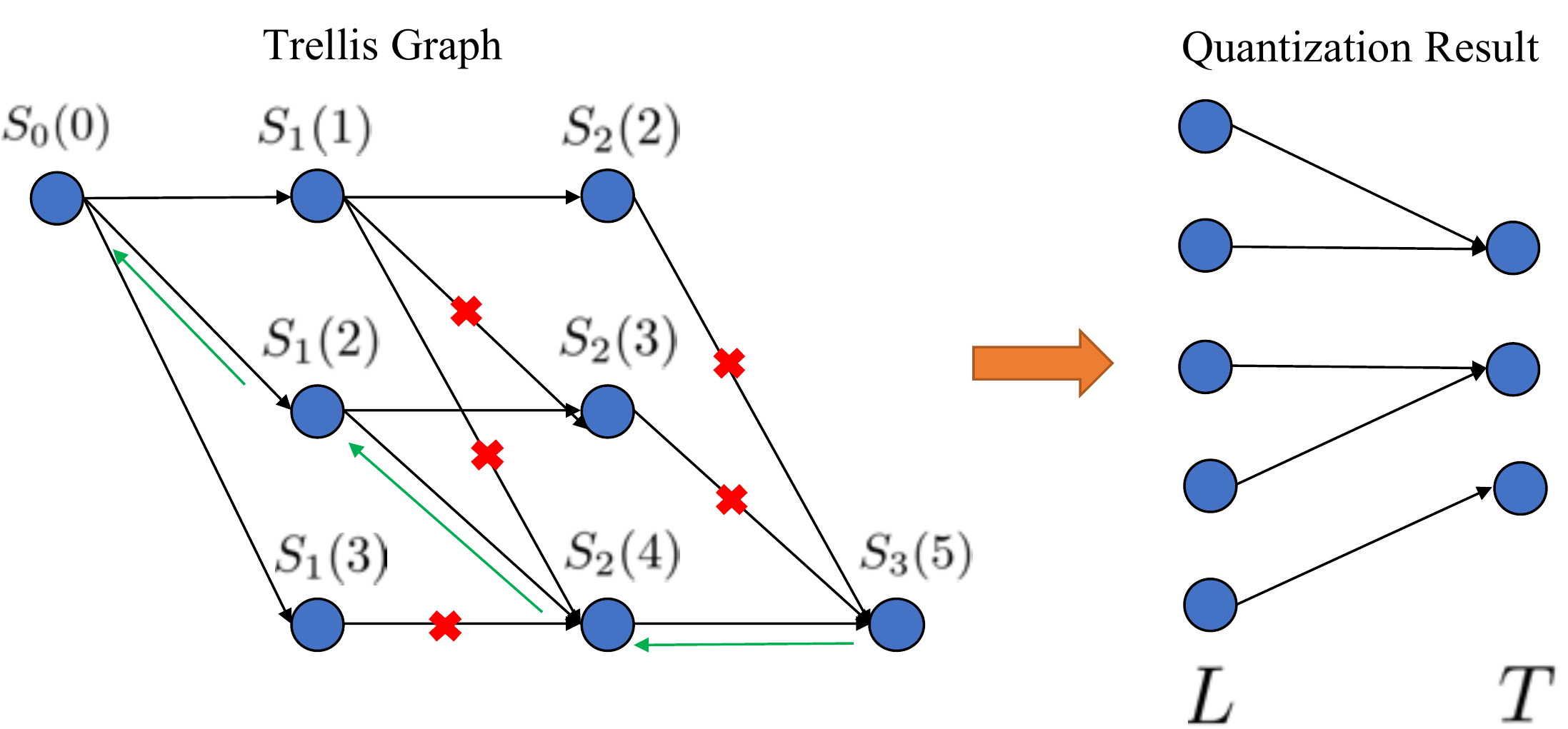}
	\caption{Illustration of the proposed minimum distortion quantization algorithm for $M=5, K=3$. The path consists of green arrows is the result of backtracing. The red crosses represent the discarded paths in the forward computation.} 
	\label{Trellis} 
	\vspace{-0.6cm}
\end{figure}
\begin{thm}
	Given a discrete probability distribution with the representative values satisfying (\ref{order_y}), Algorithm 1 produces the optimal upper boundaries $a_0 = 0 < a_1^* < a_2^* < ... < a_K = M$ and reconstrction value set $T_Q=\{t_1^*,...,t_K^*\}$.
\end{thm}
\begin{proof}
	We first emphasis that Algorithm 1 is a sequential decision process based on dynamic programming in essence and we have precomputed and stored the partial distortion and the optimal reconstruction value for all valid quantization scheme for a quantizer output (Line 3 to Line 8). With the aid of Lemma \ref{Lemma1} and the definition of state variables as welle as partial distortion, we have
	\begin{align*}
		&S_k(a_k) = \min_{a_1\sim a_k} \sum_{n=1}^{k}\sum_{l = a_{n-1}+1}^{a_n}(l-t_n)^2p(l)\\
					 &= \min_{a_k}\{\min_{a_1\sim a_{k-1}}  \sum_{n=1}^{k-1}\sum_{l = a_{n-1}+1}^{a_n}(l-t_n)^2p(l) + c(a_{k-1} \rightarrow a_k)\}\\
					 & = \min_{a_k} \{S_{k-1}(a_{k-1}) + c(a_{k-1} \rightarrow a_k)\} 
	\end{align*}
	which proves (\ref{Sk_a}). The second equation utilizes the separable property of the objective function. Note that for fixed $a_{n-1}$ and $a_{n}$, $t_n$ is obtained through (\ref{masscenter}) and most importantly, it is computed in advance and stored. Therefore, for a fixed set of upper boundaries, the quantized output values always minimize the squared error.  Assuming $a_1\sim a_k$ is the optimal upper boundaries when quantizing the source symbols $1\sim a_k$ to quantizer output $1\sim k$, we claim that $a_1\sim a_{k-1}$ is also the optimal upper boundaries for quantizing the source symbols $1\sim a_{k-1}$. Otherwise, we are capable of replacing $a_1\sim a_{k-1}$ with the optimal one, yielding a smaller squared error. This contradicts the assumption and therefore we show that the problem has \textit{optimal substructure}. It is also important to note that obtaining the optimal quantization of source symbols $1\sim a_k$ necessitates computing the optimal quantization of source symbols $1\sim a_{k-1}$ for various choices of $a_{k-1}$. Consequently, the problem has \textit{overlapping subproblems}. Since Algorithm 1 is an instance of dynamic programming,  according to Bellman's principal of optimality\cite{BellmanPrincipalOpt}, it is guaranteed to find the optimal solution. 
\end{proof}
%
\vspace{-0.5cm}
\begin{algorithm}\label{algo}
	\caption{Minimum Distortion Quantizer}
	\label{hhsa}
	\begin{algorithmic}[1]  
		\Require Input density $p(L)$, representative value set $R_L=\{l_1,...,l_M\}$, quantization level $K$.
		\Ensure  Optimal upper boundaries $a^*_0=0, a^*_1, a^*_2,...,a^*_K=K$ and reconstrction value set $T_Q=\{t_1^*,...,t_K^*\}$. 
		\State $S_0(0) = 0$, $a^*_0 = 0$
		\State /*Precompute all valid partial distortion*/
		\For {$a'\in \{1, 2, ..., M\}$}
			\For {$a \in \{a', a'+1, ..., \min\{a'+M-K, M\}\}$}
				\State $t^* = \frac{\sum_{l=a'}^{a}lp(l)}{\sum_{l=a'}^{a}p(l)}$
				\State $c(a'\rightarrow a) = \sum_{l=a'}^{a}(l-t^*)^2p(l)$
			\EndFor 
		\EndFor
		\State /*Forward computation*/
		\For {$z\in\{1, 2,...,K\}$}
			\For {$a\in \{z, z+1, ..., z+M-K\}$}
				\State $S_z(a) = \min_{a'} \ \{S_{z-1}(a') + c(a'\rightarrow a)\}$ 
				\State $d_z(a) = \arg\min_{a'} S_{z-1}(a') + c(a'\rightarrow a)$
			\EndFor
		\EndFor
		\State /*Backtracing*/
		\State $a^*_K=M$
		\For {$z \in \{K - 1, K-2, ..., 1\}$}
			\State  $a^*_z = d_{z+1}(a^*_{z+1})$
		\EndFor
	\end{algorithmic}
\end{algorithm}
\subsection{Quantized Decoding}\label{QDecode}
Given a polar code $\mathcal{P}(N,K)$ and distribution of the quantized channel LLR, we first perform the quantized density evolution as described in Section \ref{QuantizedDensityEvolution} during which we design the quantizers for replacing the $f$ and $g$ function in the conventional float-point decoder. For fixed code length, a set of quantizers for quantized decoding is designed and they are utilized in varying code rate.

 The precomputed quantizers and the correspond optimal reconstruction value set are utilized in SC and SCL decoding. Note that all quantizers can be implemented as lookup tables consist of unsigned integers and therefore all messages during decoding are simply unsigned integers as well. It is also worth emphasizing that the precomputed and stored reconstruction values of the quantized LLR symbol received by the leaf node in the decoding tree should be fetched for bit estimation (\ref{SCDecision}) or PM update (\ref{PM-HWF}).
 \vspace{-0.3cm}
\section{Simulation Results}\label{Simulation}
In this section we provide some simulation results on the error correction performance of the quantized decoders and their float-point counterparts. Performance comparison between the proposed nonuniform quantized decoder and the uniform quantized decoder\cite{UniformQuantizationSCL} is illustrated to show the superiority of the proposed scheme. Furthermore, similar to \cite{UniformQuantizationSCL}, we utilize GA to obtain the approximate distribution of internal LLRs and apply Lloyd algorithm to design quantizers, yielding another quantized decoder, and its performance is shown as well. The design metric for the all quantized decoders is to minimize squared error for fair comparison. We provide the results of  $\mathcal{P}(256, 128)$ and $\mathcal{P}(512, 128)$ as an example and similar results can be observed under other code parameters. It is important to note that all quantized decoders are generated with design $E_b/N_0=0$ dB for fair comparison.

Several interesting conclusions can be drawn from Fig. \ref{SCL-(128, 64)}. First but most important, one can immediately find out that the error correction performance of the proposed quantized decoder exceeds that of the uniform quantized decoder or the quantized decoder based on Lloyd algorithm by a large margin. For instance, approximately 0.5 dB performance gain are achievable for SC decoding with nonuniform quantized decoder under 4-bit quantization at block error rate (BLER) $10^{-2}$ and it is even better than 5-bit uniform quantized decoder or the quantized decoder with Lloyd algorithm. We conjecture that the performance degradation is due to GA itself may not provide a precise approximation for the distribution of internal LLRs, leading to a degraded performance, whereas the proposed framework does not rely on any prior assumption on the distribution of internal LLRs and all performance loss comes from quantization.  

Secondly,  from Fig. \ref{SCL-(128, 64)}, we find that the proposed quantized decoders has near float-point error correction performance compared to their float-point counterparts under 5-bit quantization. The BLER performance loss of the quantized SC decoder is neglectable at BLER $10^{-2}$ and the gap between the quantized SCL decoder and the float-point decoder is about 0.1 dB at BLER $10^{-2}$. We speculate the reason for more performance degradation is that SCL decoding relies on the accurate PM value to identify reliable decoding path. The quantized LLR decreases the resolution of PM values, making it hard to identify correct decoding path. Similar results are observed in Fig. \ref{SCL-(256, 128)} as well. 

Addtionally, when we compare Fig. \ref{SCL-(128, 64)} and Fig. \ref{SCL-(256, 128)}, one may immediately find that the quantized decoder for longer polar codes suffer from more performance loss than that for shorter polar codes. This observation is not surprising since the quantized density evolution is performed in a stage-by-stage fashion.  Quantization in a given stage will incur some distortion. Since longer polar codes correspond to more stages, more distortion will be accumulated, leading to more performance degradation. However, for current 5G downlink control channel, the maximum mother code length is 512\cite{3GPP_TS38212}, making 5-bit nonuniform quantized decoder practical in hardware implementation. 

Finally, when it comes to the robustness of the proposed quantized decoder, since it is generated with a specific design $E_b/N_0$, its robustness to varing signal-to-noise ratios can be verified by the presented simulation results. What's more, one can find that it works properly under different code length and code rate, showing its robustness to different code parameters. 
\begin{figure}[t!]
	\centering
	\includegraphics[width=.4\textwidth]{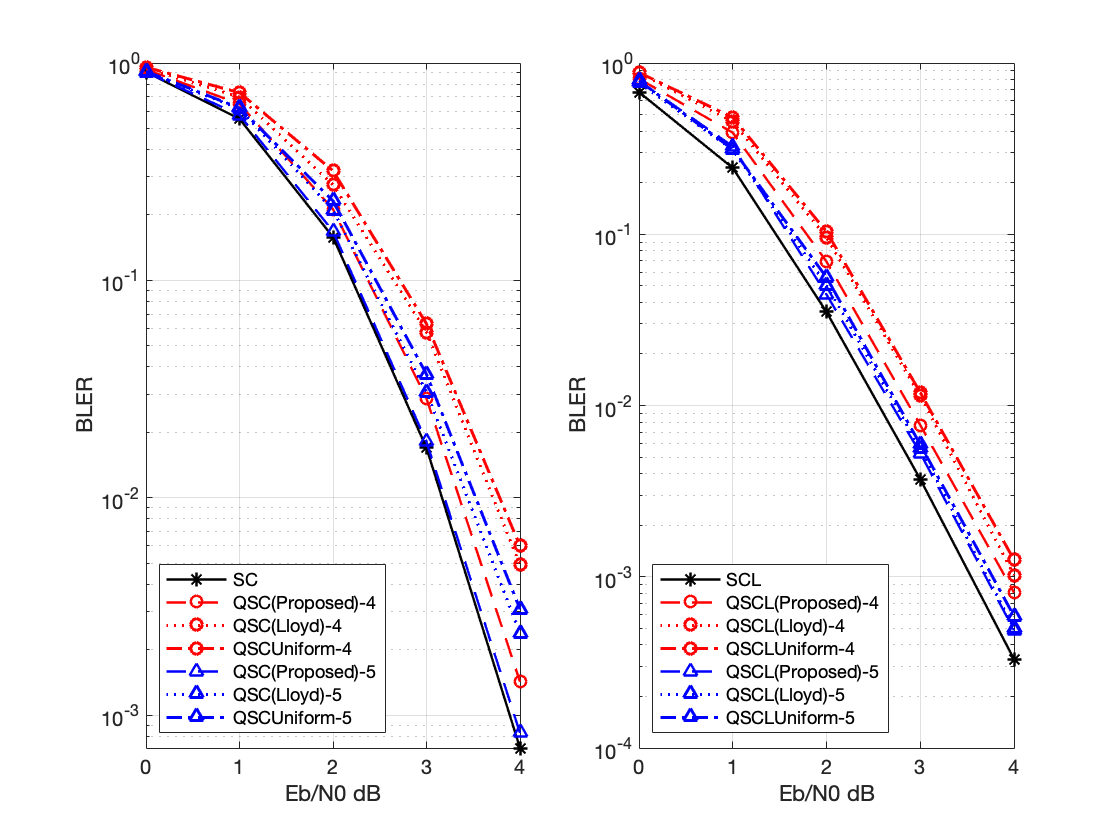}
	\caption{BLER of the proposed quantized decoders and the float-point decoders for $\mathcal{P}(256, 128)$. List size for SCL decoding is 8.} 
	\label{SCL-(128, 64)} 
	\vspace{-0.6cm}
\end{figure}
\vspace{-0.2cm}
\begin{figure}[t!]
	\centering
	\includegraphics[width=.4\textwidth]{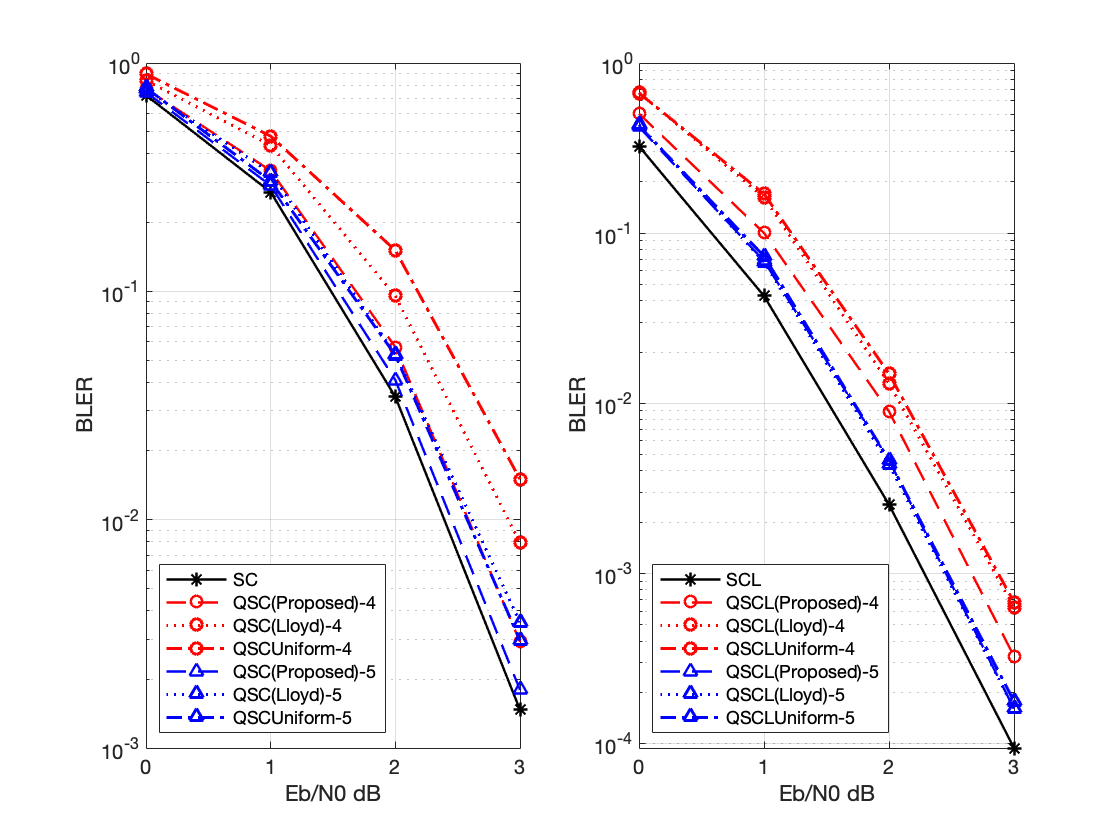}
	\caption{BLER of the proposed quantized decoders and float-point decodersr for $\mathcal{P}(512, 128)$. List size for SCL decoding is 8.} 
	\label{SCL-(256, 128)} 
	\vspace{-0.6cm}
\end{figure}

\section{Conclusion}\label{conclusion}
In this letter, we propose a nonuniform quantized decoder for polar codes with minimum distortion quantizers. We obtain the nonuniform quantizer through dynamic programming. Simulation results manifest that the proposed quantized decoder has approaching error correction performance compared to its float-point counterpart under 5-bit nonuniform quantization. Last but not least, the proposed scheme outperformed the uniform quantized decoder and the quantized decoder based on Lloyd algorithm significantly with respect to the error correction performance under the same quantization resolution. 
\section*{Acknowledgements}
We thanks Dr. Xing Yang and other researchers from Wireless Terminal Chipset Algorithm Development Dept, Hisilicon, HUAWEI TECHNOLOGIES CO., LTD for their insightful suggestions. 
\bibliographystyle{IEEEtran} 
\bibliography{references}
\end{document}